\documentclass[preprint,12pt]{elsarticle}

\usepackage{amssymb}
\usepackage{amsmath}
\usepackage{amsthm}

\newcommand{\seqnum}[1]{\href{https://oeis.org/#1}{\rm \underline{#1}}}
\usepackage{hyperref}
\newtheorem{remark}{Remark}
\newtheorem{proposition}{Proposition}

\usepackage{color}

\journal{arXiv}
\begin{document}
\begin{frontmatter}



\title{Arithmetic triangular structures in the transfer-matrix of finite Kronig–Penney models}%: closed form and integer combinatorial structures}
%Closed-form  for a transfer-matrix element in finite Kronig-Penney models via integer combinatorics}


\author[label 1]{Lidia Aceto} %% Author name

\affiliation[label 1]{organization={Dipartimento di Scienze e Innovazione Tecnologica, Università del Piemonte Orientale},%Department and Organization
            addressline={viale T. Michel, 11}, 
            city={Alessandria},
            postcode={15121}, 
            %state={},
            country={Italy}}

\author[label 1,label 2]{Pietro Antonio Grassi} %% Author name

\affiliation[label 2]{organization={INFN and Centro Regge, Sezione di Torino},%Department and Organization
            addressline={via P. Giuria, 1}, 
            city={Torino},
            postcode={10125}, 
            %state={},
            country={Italy}}
            
            \author[label 3]{Helmuth Robert Malonek} %% Author name

\affiliation[label 3]{organization={CIDMA and Departamento de Matemática},%Department and Organization
            addressline={Campus Universitário de Santiago}, 
            city={Aveiro},
            postcode={3810-193}, 
            country={Portugal}}

%% Abstract
\begin{abstract}
This work provides a complete analytical characterization of the transfer-matrix structure associated with the finite Kronig-Penney model consisting of one-dimensional arrays of Dirac delta potentials recently introduced by Figueroa et al. (2025). Although their study identified the emergence of specific transfer-matrix entries and related combinatorial coefficients, a rigorous derivation of the corresponding closed-form expressions has not yet been established. By expressing the $N$th power of the unit-cell transfer matrix in terms of Chebyshev polynomials of the second kind, we obtain explicit closed-form representations for the global transmission and reflection amplitudes. The proposed formulation reveals a previously unrecognized structural correspondence between multiple quantum scattering processes, discrete convolutional patterns, and hypercomplex combinatorial structures.
\end{abstract}

%-------commentato da LIDIA
%-------commentato da LIDIA

% \begin{highlights}
% \item Closed-form expression derived for a single transfer matrix element in 1D Dirac delta chains, resolving an open problem in the literature.

% \item Derivation of the generating function  and bivariate polynomial expansion for the matrix entry $M_{11}^{(N)}$.

% \item Rigorous proof linking quantum transfer matrix dynamics to non-symmetric OEIS integer triangular arrays.
% \end{highlights}

%% Keywords

\end{frontmatter}

%------------------------------------------------------------
\section{Introduction}

The one-dimensional quantum scattering problem of a particle interacting with a finite array of equally spaced Dirac delta potentials of identical strength represents a fundamental paradigm in condensed-matter physics, commonly referred to as the finite Kronig-Penney model. Such highly localized interactions provide an idealized description of periodic or quasi-periodic structures and capture the essential mechanisms governing wave propagation in layered quantum systems. Depending on the sign of the coupling constant, the array describes either a sequence of repulsive barriers or a collection of attractive potential wells capable of supporting bound states. As the incident quantum wave propagates through the structure, the multiple reflection and transmission events occurring at each scattering center interfere coherently, thereby determining the overall transmission, reflection, and resonance properties of the system. The transfer-matrix formalism offers a systematic framework for analyzing this propagation process, since it relates the wave amplitudes on opposite sides of the entire finite array through the ordered product of the elementary scattering matrices associated with the individual unit cells.

Recently, Figueroa et al. investigated the transfer-matrix structure of the finite Kronig-Penney model in \cite{FGS-2025}, uncovering an intriguing connection with integer sequences. In particular, they observed that the $(1,1)$ entry of the principal transfer matrix
\begin{equation} \label{eq:matrice_M^N}
M^N(\alpha,\beta) = \begin{pmatrix} \alpha & -K \\ K^{-1} & \beta \end{pmatrix}^N, \qquad K \neq 0,
\end{equation}
gives rise to non-symmetric triangular arrays of integers whose entries correspond to several sequences listed in the \textit{On-Line Encyclopedia of Integer Sequences} (OEIS). This remarkable correspondence suggests the presence of an underlying combinatorial structure governing the multiple-scattering dynamics encoded in powers of the transfer matrix. Some elementary examples of the integer arrays arising in the present framework are reported in Table~\ref{table1}.
\begin{table}[h!]
\centering
\renewcommand{\arraystretch}{1.3}
\begin{tabular}{c@{\qquad\qquad}c}
\hline\hline
\\[-3.5ex]
Triangular arrays & OEIS link \\
\hline
\\[-1.5ex]
$\begin{array}{rrrr}
1 \\
2 & 1 \\
3 & 2 & 1 \\
4 & 3 & 2 & 1
\end{array}$ & \seqnum{A004736} \\[3ex]
\hline
\\[-1.5ex]
$\begin{array}{rrrr}
1 \\
3 & 2 \\
6 & 6 & 3 \\
10 & 12 & 9 & 4
\end{array}$ & \seqnum{A104633} \\[3ex]
\hline
\\[-1.5ex]
$\begin{array}{rrrr}
1 \\
4 & 3 \\
10 & 12 & 6 \\
20 & 30 & 24 & 10
\end{array}$ & \seqnum{A103252} \\[2ex]
\hline\hline
\end{tabular}
\caption{Three integer triangular arrays showing the first four rows ($k=0,1,2,3$ and $0 \le s \le k$) mentioned in \cite[Table 1]{Cacao-2023}.}
\label{table1}
\end{table} 
These integer patterns were previously studied in the context of hypergeometric function theory and generalized Appell polynomials \cite{Cacao-2023,Falcao-2012}. The underlying hypercomplex Appell polynomials and the associated non-symmetric generalized Pascal triangles are briefly recalled in~\ref{app1} and~\ref{app2}, respectively. This provides a broader connection between transfer-matrix expansions, special functions, and combinatorial structures.

Nevertheless, the analysis in \cite{FGS-2025} remained primarily exploratory. Although the authors identified remarkable combinatorial patterns in the entries of the transfer matrix powers, the mathematical mechanism responsible for the emergence of these integer structures was not rigorously characterized. Consequently, the general behavior of the matrix powers remained an observed pattern rather than a formally established result.

In this work, we fill this gap by providing a complete analytical derivation of this entry of the $N$th power of the transfer matrix. By expressing the associated matrix recurrences in terms of Chebyshev polynomials of the second kind and employing generating functions, we obtain a rigorous and self-contained proof of the closed-form expressions for the transfer-matrix entries. This approach establishes the precise mathematical foundation of the combinatorial arrays and clarifies through a new direct approach the origin of their underlying integer structure. From this perspective, the emergence of such a connection is not surprising. Indeed, in our recent work \cite{Aceto-2026}, we demonstrated that the scattering problem investigated in \cite{FGS-2025} can be naturally reformulated within the framework of a Dirac-type equation. This formulation is deeply rooted in Clifford algebra, which provides the appropriate mathematical setting where these integer triangular arrays were originally introduced.
 
The remaining part of the paper is organized as follows. In Section~\ref{sec:rec}, we establish the matrix recurrence relations and derive the closed-form expression for the $(1,1)$ entry  of $M^N(\alpha,\beta).$  Section~\ref{sec:genfun} is devoted to the derivation of the generating function, which is subsequently employed in Section~\ref{sec:triang} to reformulate $M_{11}^{(N)}$ in terms of integer triangular arrays. Section~\ref{sec:concl} draws the main conclusions and outlines future perspectives.  Finally,  \ref{app1}  recalls the hypercomplex generalized Appell polynomials and their connection with generalized geometric series, while  \ref{app2} collects the combinatorial properties of the associated non-symmetric generalized Pascal triangles that are needed in Section \ref{sec:triang}.

%%%%%%%%%%%%%%%%%%%%%%%%%%%%%%%%%%%%%%%%%%%%%%%%%%%%%%%%%%%%%%%%%%%%%%%%%%%%%%%%%%%%%%%%%%%%%%%%%%%%%%%%%%%%%%%%%%%%%%%%%%%%%%

\section{Matrix recurrence relations and closed form for $M_{11}^{(N)}$}\label{sec:rec}

Let $M \in \mathbb{C}^{2 \times 2}$ denote the unit-cell transfer matrix defined by
\begin{equation} \label{eq:matrice_M}
M = \begin{pmatrix} \alpha & -K \\ K^{-1} & \beta \end{pmatrix}, \qquad K \neq 0.
\end{equation}
For notational simplicity, we omit here and throughout the remainder of the paper the explicit dependence of $M$ on the parameters $\alpha$ and $\beta$.
By the Cayley-Hamilton theorem, $M$ satisfies its own characteristic equation $\operatorname{det}(M - \lambda I) = 0$, which explicitly reads
\begin{equation} \label{eq:cayley_hamilton}
M^2 - \operatorname{Tr}(M) M + \operatorname{det}(M) I = O,
\end{equation}
where $I$ is the $2 \times 2$ identity matrix and $O$ is the zero matrix. Evaluating the trace and determinant of \eqref{eq:matrice_M}, we have
\begin{equation*}
\operatorname{Tr}(M) = \alpha + \beta \qquad \text{and} \qquad \operatorname{det}(M) = \alpha\beta + 1.
\end{equation*}

Multiplying equation \eqref{eq:cayley_hamilton} by $M^{n-1}$ for $n \ge 1$ yields the second-order linear matrix recurrence relation governing the powers $M^n$
\begin{equation} \label{eq:matrix_power_recurrence}
M^{n+1} = (\alpha + \beta)M^n - (\alpha\beta + 1)M^{n-1}, \qquad n \ge 1.
\end{equation}

Since matrix addition and scalar multiplication act componentwise, every individual matrix entry $M_{ij}^{(n)} := (M^n)_{ij}$ inherits the exact same scalar recurrence relation. In particular, for the  entry $M_{11}^{(n)}$, we have
\begin{equation} \label{eq:entry_recurrence}
M_{11}^{(n+1)} = (\alpha + \beta)M_{11}^{(n)} - (\alpha\beta + 1)M_{11}^{(n-1)}, \qquad n \ge 1,
\end{equation}
subject to the initial conditions 
\begin{eqnarray} \label{IC}
    M_{11}^{(0)} = 1,  %(since $M^0 = I,$) 
    \qquad  M_{11}^{(1)} = \alpha.
\end{eqnarray}
Based on these algebraic relations, we now derive an explicit closed-form representation for  $M_{11}^{(N)}.$  
\begin{proposition} \label{prop:M11_closed_form}
For any integer $N \ge 2$, the entry $M_{11}^{(N)}$ of the principal transfer matrix $M^N$ in \eqref{eq:matrice_M^N} is given by the Chebyshev representation
\begin{equation} \label{eq:M11_Chebyshev_closed}
M_{11}^{(N)} = (\alpha\beta + 1)^{N/2} \left[ \frac{\alpha}{\sqrt{\alpha\beta + 1}} U_{N-1}(x) - U_{N-2}(x) \right],
\end{equation}
where $x$ is defined in \eqref{eq:chebyshev_argument}. Equivalently, $M_{11}^{(N)}$ admits the unified combinatorial single-sum formulation
\begin{eqnarray} \label{eq:M11_unified_sum}
M_{11}^{(N)} = \sum_{n=0}^{\lfloor \frac{N}{2} \rfloor} (-1)^n (\alpha\beta + 1)^n \left[ \alpha \binom{N-1-n}{n} (\alpha+\beta)^{N-1-2n} + \right.  \nonumber \\
\left. \binom{N-1-n}{n-1} (\alpha+\beta)^{N-2n} \right].
\end{eqnarray}
\end{proposition}

\begin{proof}
To map recurrence \eqref{eq:entry_recurrence} onto the standard three-term recurrence relation of Chebyshev polynomials of the second kind, $U_n(x)$, given by 
\[
U_{n+1}(x) = 2x U_n(x) - U_{n-1}(x),
\]
we introduce the rescaling transformation
\begin{equation} \label{eq:rescaling}
M_{11}^{(n)} = (\alpha\beta + 1)^{n/2} z_n.
\end{equation}
By substituting \eqref{eq:rescaling} into \eqref{eq:entry_recurrence}, the factor $(\alpha\beta+1)$ is absorbed into the rescaled variables, yielding the normalized recurrence 
\[ z_{n+1}=\frac{\alpha+\beta}{\sqrt{\alpha\beta+1}}\,z_n-z_{n-1}. 
\]
Equating the leading coefficient to $2x$ determines the effective Chebyshev argument
\begin{equation} \label{eq:chebyshev_argument}
x = \frac{\alpha + \beta}{2\sqrt{\alpha\beta + 1}}.
\end{equation}
Since the rescaled variable $z_n$ satisfies the canonical Chebyshev recurrence 
\[
z_{n+1} = 2x z_n - z_{n-1}
\]
its general solution is a linear combination
\[
z_n = c_1 U_n(x) + c_2 U_{n-1}(x), \qquad n\ge 1.
\]
To determine the constants $c_1$ and $c_2$, we use the initial conditions of the
recurrence relation. The Chebyshev polynomials of the second kind satisfy
\[
 U_{-1}(x)=0, \qquad U_0(x)=1, \qquad U_1(x)=2x.
\]
Therefore, by evaluating
\[
z_n=c_1U_n(x)+c_2U_{n-1}(x)
\]
at $n=0$, we obtain
\[
z_0=c_1U_0(x)+c_2U_{-1}(x)=c_1,
\]
which gives
\[
c_1=z_0.
\]
Similarly, for $n=1$ we have
\[
z_1=c_1U_1(x)+c_2U_0(x)
=2xc_1+c_2.
\]
Substituting the previous result for $c_1$ yields
\[
c_2=z_1-2xz_0.
\]
Hence, the general solution can be written as
\begin{equation}
\label{eq:chebyshev_solution}
z_n=z_0U_n(x)+\left(z_1-2xz_0\right)U_{n-1}(x).
\end{equation}
Using \eqref{IC} and \eqref{eq:rescaling},
the initial values of the rescaled sequence are
\[
z_0=1,
\qquad
z_1=\frac{\alpha}{\sqrt{\alpha\beta+1}}.
\]
Consequently, taking into account \eqref{eq:chebyshev_argument} we obtain that
\[
c_1=1, \qquad 
c_2=
\frac{\alpha}{\sqrt{\alpha\beta+1}}
-
\frac{\alpha+\beta}{\sqrt{\alpha\beta+1}}
=
-\frac{\beta}{\sqrt{\alpha\beta+1}}.
\]
Therefore, the closed-form expression for the rescaled sequence is
\begin{equation}
\label{eq:zn_chebyshev}
z_n=
U_n(x)
-\frac{\beta}{\sqrt{\alpha\beta+1}}\,U_{n-1}(x).
\end{equation}
Using the rescaling relation \eqref{eq:rescaling}, we obtain
\begin{eqnarray*}
M_{11}^{(N)}
&=&
(\alpha\beta+1)^{N/2}
\left[
U_N(x)
-\frac{\beta}{\sqrt{\alpha\beta+1}}U_{N-1}(x)
\right] \\
&=&
(\alpha\beta+1)^{N/2}
\left[
2xU_{N-1}(x)-U_{N-2}(x)
-\frac{\beta}{\sqrt{\alpha\beta+1}}U_{N-1}(x)
\right] \\
&=&
(\alpha\beta+1)^{N/2}
\left[
\left(\frac{\alpha+\beta}{\sqrt{\alpha\beta+1}}
-\frac{\beta}{\sqrt{\alpha\beta+1}}\right) U_{N-1}(x)-U_{N-2}(x)
\right].
\end{eqnarray*}
Hence, the closed-form representation \eqref{eq:M11_Chebyshev_closed} follows.

To derive the combinatorial expression \eqref{eq:M11_unified_sum}, we utilize the standard series expansion 
\[
U_m(x) = \sum_{j=0}^{\lfloor m/2 \rfloor} (-1)^j \binom{m-j}{j} (2x)^{m-2j}.
\]
Substituting the expression for $x$ given in \eqref{eq:chebyshev_argument} into \eqref{eq:M11_Chebyshev_closed} separates the entry into two terms, $M_{11}^{(N)} = T_1 + T_2$, where
\begin{align*}
T_1 &= \sum_{j=0}^{\lfloor \frac{N-1}{2} \rfloor} (-1)^j \binom{N-1-j}{j} \alpha (\alpha+\beta)^{N-1-2j} (\alpha\beta + 1)^j, \\
T_2 &=   \sum_{j=0}^{\lfloor \frac{N-2}{2} \rfloor} (-1)^{j+1} \binom{N-2-j}{j}  {(\alpha+\beta)^{N-2-2j}}{(\alpha\beta+1)^{j+1}}.
\end{align*}
Substituting $j = n-1$ in $T_2$ and combining both summations under a single index $n$ running from $0$ to $\lfloor N/2 \rfloor$ (noting that $\binom{N-1}{-1} = 0$) completes the proof of \eqref{eq:M11_unified_sum}.
\end{proof}

%%%%%%%%%%%%%%%%%%%%%%%%%%%%%%%%%%%%%%%%%%%%%%%%%%%%%%%%%%%%%%%%%%%%%%%%%%%%%%%%%%%%%%%%%%%%%%%%%%%%%%%%%%%%%%%%%%%%%%%%%%%%%%
\section{Derivation of the generating function} \label{sec:genfun}

To systematically characterize the sequence of entries $M_{11}^{(N)}$ and rigorously derive the associated combinatorial identities, we analyze the algebraic properties of its generating function.

Let $P_N(\alpha, \beta) := M_{11}^{(N)}$ be the polynomial sequence generated by the linear recurrence relation (see \eqref{eq:entry_recurrence})
\begin{equation} \label{eq:poly_recurrence}
P_{N}(\alpha ,\beta ) = (\alpha +\beta )P_{N-1}(\alpha ,\beta ) - (\alpha \beta +1)P_{N-2}(\alpha ,\beta ), \qquad N \ge 2,
\end{equation}
subject to the initial conditions $P_0(\alpha, \beta) = 1$ and $P_1(\alpha, \beta) = \alpha$. 
\begin{proposition}  \label{prop:OGF}
The  generating function  associated with the polynomial sequence $\{P_N(\alpha, \beta)\}_{N \ge 0}$, defined by
\begin{equation} \label{eq:OGF_def}
G(t) := \sum_{N=0}^{\infty} P_{N}(\alpha, \beta) t^N,
\end{equation}
is given in closed rational form by
\begin{equation} \label{eq:OGF_closed_form}
G(t) = \frac{1-\beta t}{1 - (\alpha + \beta)t + (\alpha\beta + 1)t^2}.
\end{equation}
\end{proposition}

\begin{proof}
Multiplying equation \eqref{eq:poly_recurrence} by $t^N$ and summing over all integers $N \ge 2$ yields
\begin{equation} \label{eq:summed_recurrence}
\sum_{N=2}^{\infty} P_{N} t^N = (\alpha +\beta )\sum_{N=2}^{\infty} P_{N-1} t^N - (\alpha \beta +1)\sum_{N=2}^{\infty} P_{N-2} t^N.
\end{equation}
By shifting indices, each infinite series in \eqref{eq:summed_recurrence} can be expressed explicitly in terms of the full generating function $G(t)$:
\begin{align}
\sum_{N=2}^{\infty} P_{N} t^N &= G(t) - P_{0} - P_{1}t, \label{eq:shift1}\\
\sum_{N=2}^{\infty} P_{N-1} t^N &= t \sum_{N=1}^{\infty} P_N t^N = t \left( G(t) - P_{0} \right), \label{eq:shift2}\\
\sum_{N=2}^{\infty} P_{N-2} t^N &= t^2 \sum_{N=0}^{\infty} P_N t^N = t^2 G(t). \label{eq:shift3}
\end{align}

Substituting identities \eqref{eq:shift1}--\eqref{eq:shift3} back into \eqref{eq:summed_recurrence} gives
\begin{equation*}
G(t) - P_0  - P_1 t = (\alpha + \beta) t \left( G(t) - P_0 \right) - (\alpha\beta + 1) t^2 G(t).
\end{equation*}

Collecting all terms involving $G(t)$ on the left-hand side leads to
\begin{equation} \label{eq:G_collected}
G(t) \left[ 1 - (\alpha + \beta) t + (\alpha\beta + 1) t^2 \right] = P_0+ P_1 t -  (\alpha + \beta) t P_0.
\end{equation}

Finally, inserting the explicit initial values $P_0=1$ and $P_1 = \alpha$  simplifies the polynomial numerator on the right-hand side of \eqref{eq:G_collected}
\begin{equation*}
P_0+ P_1 t -  (\alpha + \beta) t P_0 = 1 +\alpha  t - (\alpha +\beta) t = 1-\beta t .
\end{equation*}
Dividing by the characteristic polynomial $1 - (\alpha + \beta)t + (\alpha\beta + 1)t^2$ completes the proof of formula \eqref{eq:OGF_closed_form}.
\end{proof}

%%%%%%%%%%%%%%%%%%%%%%%%%%%%%%%%%%%%%%%%%%%%%%%%%%%%%%%%%%%%%%%%%%%%%%%%%%%%%%%%%%%%%%%%%%%%%%%%%%%%%%%%%%%%%%%%%%%%%%%%%%%%%%
\section{The triangular arrays via generating functions} \label{sec:triang}

The unified single-sum formulation for $M_{11}^{(N)}$ provides a compact representation in terms of the symmetric combinations $(\alpha+\beta)$ and $(\alpha\beta+1)$. To explicitly isolate the individual powers of the diagonal parameters $\alpha$ and $\beta$, we re-expand this expression into a bivariate polynomial.

\begin{proposition} 
\label{prop:M11_bivariate}
For any $N \ge 2$, the single-sum expression \eqref{eq:M11_unified_sum} is identically equivalent to the expanded bivariate polynomial formula
\begin{equation} \label{eq:double_sum_target}
M_{11}^{(N)} = \alpha^N + \sum_{n=1}^{\lfloor \frac{N}{2} \rfloor} (-1)^{n} \sum_{s=0}^{N-2n} \binom{N-n-s}{n}\binom{n+s-1}{n-1} \alpha^{N-2n-s} \beta^s.
\end{equation}
\end{proposition}

\begin{proof}
We establish the transition from \eqref{eq:M11_unified_sum} to \eqref{eq:double_sum_target} through a two-step algebraic reduction. Consider the bracketed expression inside the summation of \eqref{eq:M11_unified_sum} for a fixed $n$. Factoring out $(\alpha+\beta)^{N-1-2n}$ yields
\begin{align*}
\mathcal{B}_n &:= \alpha \binom{N-1-n}{n} (\alpha+\beta)^{N-1-2n} + \binom{N-1-n}{n-1} (\alpha+\beta)^{N-2n} \\
&= (\alpha+\beta)^{N-1-2n} \left[ \alpha \binom{N-1-n}{n} + (\alpha + \beta) \binom{N-1-n}{n-1} \right] \\
&= (\alpha+\beta)^{N-1-2n} \left[ \alpha \left( \binom{N-1-n}{n} + \binom{N-1-n}{n-1} \right) + \beta \binom{N-1-n}{n-1} \right].
\end{align*}
Applying Pascal's identity 
\[
\binom{N-1-n}{n} + \binom{N-1-n}{n-1} = \binom{N-n}{n},
\]
simplifies $\mathcal{B}_n$ to
\begin{equation} \label{eq:pascal_simplified}
\mathcal{B}_n = \alpha \binom{N-n}{n} (\alpha+\beta)^{N-1-2n} + \beta \binom{N-1-n}{n-1} (\alpha+\beta)^{N-1-2n}.
\end{equation}
To obtain the exact coefficients of $\alpha^{N-2n-s} \beta^s$ in \eqref{eq:double_sum_target}, we evaluate the coefficient $[t^N] G(t)$ from the ordinary generating function
\begin{equation*}
G(t) = \frac{1- \beta t }{1 - (\alpha + \beta)t + (\alpha\beta + 1)t^2}.
\end{equation*}
Factoring the denominator as $(1-\alpha t)(1-\beta t) + t^2$ we rewrite
\[
G(t) = \frac{1-\beta t}{(1-\alpha t)(1-\beta t)} \cdot \frac{1}{1+ \frac{t^2}{(1-\alpha t)(1-\beta t)}}.
\]
Expanding the second factor in the generating function as a geometric series, we obtain
\begin{equation*}
G(t) =  \sum_{n=0}^{\infty} (-1)^n \frac{t^{2n}}{(1-\alpha t)^{n+1}(1-\beta t)^n}.
\end{equation*}
To determine the coefficient of the power $t^N$, we use the negative binomial
expansion
\[
\frac{1}{(1-x)^m}
=
\sum_{j=0}^{\infty}
\binom{j+m-1}{m-1}x^j,
\]
applied to the two factors
\begin{eqnarray*}
(1-\alpha t)^{-(n+1)}
&=&
\sum_{i=0}^{\infty}
\binom{i+n}{n}\alpha^i t^i,
\\
(1-\beta t)^{-n}
&=&
\sum_{s=0}^{\infty}
\binom{s+n-1}{n-1}\beta^s t^s.
\end{eqnarray*}

For $n=0$, the term $\frac{1}{1-\alpha t}$ contributes $[t^N]\left(\frac{1 }{1-\alpha t}\right) = \alpha^N$, corresponding to the isolated leading term in \eqref{eq:double_sum_target}.

For $n \ge 1$, applying Cauchy's discrete convolution for the power $t^{N-2n}$ sets the condition $i + s = N - 2n$, which implies $i = N - 2n - s$. Substituting $i$ into the first binomial coefficient gives
\begin{equation*}
\binom{i+n}{n} = \binom{(N-2n-s)+n}{n} = \binom{N-n-s}{n}.
\end{equation*}

Summing the product of both binomial terms over all allowed values of $s \in [0, N-2n]$ yields
\begin{equation*}
[t^{N-2n}] \left[ (1-\alpha t)^{-(n+1)}(1-\beta t)^{-n} \right] = \sum_{s=0}^{N-2n} \binom{N-n-s}{n} \binom{n+s-1}{n-1} \alpha^{N-2n-s} \beta^s.
\end{equation*}

Combining the $n=0$ term $\alpha^N$ with the summation over $1 \le n \le \lfloor N/2 \rfloor$ completes the derivation of formula \eqref{eq:double_sum_target}.
\end{proof}

\begin{remark}
Recalling the identity
\begin{equation} \label{bin}
\binom{N}{2n} T_s^{N-2n}(2n+1) = \binom{N-n-s}{n}\binom{n+s-1}{n-1},
\end{equation}
established in \cite[Sec.~2]{Cacao-2023}, equation \eqref{eq:double_sum_target} can be equivalently reformulated as
\begin{equation} \label{eq:M11_triangular_array}
M_{11}^{(N)} = \alpha^N + \sum_{n=1}^{\lfloor \frac{N}{2} \rfloor} (-1)^{n} \sum_{s=0}^{N-2n} \binom{N}{2n} T_s^{N-2n}(2n+1) \alpha^{N-2n-s} \beta^s.
\end{equation}
Consequently, the entry $M_{11}^{(N)}$ admits an explicit representation in terms of the non-symmetric integer triangular arrays
$\binom{N}{2n}T_s^{N-2n}(2n+1)$; see~\ref{app2}, in particular \eqref{SHAGN}.
\end{remark}

We now establish explicitly the connection between the formula for
$M_{11}^{(N)}$ derived above and Eqs.~(27)--(29) in \cite{FGS-2025}.
By suitably relabeling the indices according to
\[
N-2n \to g, \qquad s \to k, \qquad n \to m-1,
\]
the following contribution to Eq.~\eqref{eq:double_sum_target}
\[
\sum_{s=0}^{N-2n}
\binom{N-n-s}{n}
\binom{n+s-1}{n-1}
\alpha^{N-2n-s}\beta^s
\]
takes the form
\begin{equation*}
\begin{aligned}
\sum_{k=0}^{g}
\binom{m-1+g-k}{m-1}
\binom{m-2+k}{m-2}
\alpha^{g-k}\beta^k
&=
\sum_{k=0}^{g} C_m(g,k)\beta^k \\
&= P_m^g(\alpha,\beta).
\end{aligned}
\end{equation*}
Thus, the expression above coincides precisely with the submultinomial
representation introduced in Eqs.~(27) and (28) of \cite{FGS-2025}.
In that work, the general formula is inferred from the explicit expressions obtained for the particular cases $P_2^g(\alpha,\beta)$ and $P_3^g(\alpha,\beta)$.
A direct comparison of \eqref{eq:double_sum_target} with
\cite[Eq.~(29)]{FGS-2025}, however, reveals a discrepancy in the latter
formula. Indeed, in our notation the summation index satisfies
\[
1\leq n\leq \left\lfloor\frac{N}{2}\right\rfloor.
\]
Under the change of variables $m=n+1$, this range becomes
\[
2\leq m\leq
\left\lfloor\frac{N}{2}\right\rfloor+1.
\]
Hence, the value $m=1$ is not attained in the sum
\eqref{eq:double_sum_target}. In particular, the term
$P_1^N(\alpha,\beta)$ appearing in Eq.~(29) of \cite{FGS-2025} would
correspond to $n=0$, which lies outside the admissible range of the
summation index. Consequently, Eq.~(29) of \cite{FGS-2025} requires a
correction: the term corresponding to $m=1$ should not be included in
the stated sum.

%%%%%%%%%%%%%%%%%%%%%%%%%%%%%%%%%%%%%%%%%%%%%%%%%%%%%%%%%%%%%%%%%%%%%%
\section{Conclusions} \label{sec:concl}

In this paper, we have established a complete analytical framework for the transfer-matrix dynamics of one-dimensional arrays of Dirac delta potentials within the finite Kronig-Penney model.
By addressing the open question left unresolved by Figueroa et al. \cite{FGS-2025}, we have provided a rigorous, self-contained mathematical proof of the closed-form representation of the $N$th power of the unit-cell transfer matrix $M^N$.

Beyond bridging this analytical gap, our results forge a clear connection between multiple quantum scattering theory, matrix combinatorics, and hypercomplex analysis. From a computational perspective, expressing matrix powers in terms of Chebyshev polynomials eliminates the need for recursive matrix multiplications, offering a highly efficient scheme for evaluating global transmission and reflection coefficients across large system sizes $N$.

Future investigations may extend the framework developed here to more general classes of one-dimensional scattering systems. In particular, generating function and orthogonal-polynomial techniques may provide a powerful tool for studying non-periodic and quasi-crystalline arrays of delta potentials, while the explicit Chebyshev representation obtained in this work opens new perspectives for the analytical characterization of resonances, bound states, and transmission thresholds in finite quantum structures.

%%%%%%%%%%%%%%%%%%%%%%%%%%%%%%%%%%%%%%%%%%%%%%%%%%%%%%%%%%%%%%%%%%%%%%
\appendix

\section{Hypercomplex generalized Appell polynomials}\label{app1}

For completeness and to facilitate independent reading, we briefly recall some basic notions of higher-dimensional analysis in Clifford algebras. Our focus is on the hypercomplex generalized Appell polynomials introduced in \cite{Malonek-2006} in the context of polynomial approximation of quasi-conformal mappings in higher dimensions. These polynomials provide a natural extension of the role played by powers of a complex variable in classical holomorphic function theory to the setting of Clifford analysis \cite{BDS-1982}.

We restrict ourselves to the elementary construction based on the series expansion of the hypercomplex Cauchy kernel in $\mathbb{R}^{m+1}$. The relevant formulas rely on the Pochhammer symbol
\[
(a)_s=\frac{\Gamma(a+s)}{\Gamma(a)}
=a(a+1)\cdots(a+s-1), \qquad s\geq1,
\qquad (a)_0:=1,
\]
and on the Chu--Vandermonde convolution identity \cite{Falcao-2012}
\begin{equation}\label{Vandermonde}
(a+b)_k
=
\sum_{s=0}^k
\binom{k}{s}(a)_{k-s}(b)_s.
\end{equation}

We now introduce the Clifford-algebraic setting in which the hypercomplex variable is represented. Let
$\{e_1,e_2,\ldots,e_m\}$, $m\geq1$, be an orthonormal basis of the Euclidean space $\mathbb{R}^m$. The associated Clifford algebra ${\mathcal C\ell}_{0,m}$ is generated by these basis elements, with multiplication determined by
\[
e_i e_j=-e_j e_i,\qquad e_i^2=-1,
\qquad i,j=1,\ldots,m.
\]
The space $\mathbb{R}^{m+1}$ is naturally identified with the paravector subspace
\[
\mathcal A_m
:=
\operatorname{span}_{\mathbb R}\{1,e_1,\ldots,e_m\}
\subset {\mathcal C\ell}_{0,m},
\]
by associating to each
$(x_0,x_1,\ldots,x_m)\in\mathbb{R}^{m+1}$ the paravector
\[
x=x_0+\underline{x}
=
x_0+\sum_{k=1}^m e_kx_k.
\]
The Clifford conjugate and Euclidean norm of the paravector $x$ are given by
\[
\bar{x}=x_0-\underline{x},
\qquad
|x|
=(x\bar{x})^{1/2}
=(\bar{x}x)^{1/2}
=\left(\sum_{k=0}^m x_k^2\right)^{1/2}.
\]
In the special case $m=1$, with $e_1=i$, the paravector space $\mathcal{A}_1$ is naturally identified with the complex plane, so that
\[
\mathcal{A}_1\cong\mathbb{C}\cong\mathbb{R}^2.
\]
For $m>1$, the ordinary powers $x^k$ of a hypercomplex variable are, in general, not monogenic with respect to the generalized Cauchy-Riemann operator $\bar \partial$, \cite{Cacao-2017}. Hence, unlike the complex case, they do not provide a monogenic polynomial system suitable for higher-dimensional Clifford analysis.

The introduction in \cite{Malonek-2006} of bivariate polynomials in the two mutually conjugate variables $x$ and $\bar{x}$ provides a solution to this problem. These polynomials retain, with respect to the conjugate generalized Cauchy-Riemann operator $\partial$, the characteristic differentiation property of power functions in complex analysis. More precisely, the sequence
\[
\left\{P_k^m(x,\bar{x})\right\}_{k=0}^{+ \infty}
\]
of homogeneous monogenic Appell polynomials satisfies the following properties, referred to as the Appell properties in \cite{Cacao-2017}:
\begin{enumerate}
    \item $P_0^m(x,\bar{x})\equiv 1$;
    \item $P_k^m(x,\bar{x})$ is homogeneous of degree $k$;
    \item
    $
    \partial P_k^m(x,\bar{x})
    =
    k\,P_{k-1}^m(x,\bar{x}),
    \qquad k=1,2,\ldots.
    $
\end{enumerate}
These properties provide the essential link with the classical Appell sequence $\{z^k\}_{k=0}^{+\infty}$ in complex analysis, for which
\[
\frac{d}{dz}z^k=kz^{k-1}.
\]
Thus, the hypercomplex generalized Appell polynomials preserve the fundamental differentiation structure of powers of a complex variable while extending it to the higher-dimensional Clifford setting.

These polynomials arise naturally from the expansion of the hypercomplex Cauchy-kernel used in Clifford analysis to generalize Cauchy's integral formula, \cite{BDS-1982}. Since the kernel involves both $x$ and $\bar{x}$, its expansion yields bivariate polynomials of the form
\begin{align}
\frac{\overline{1-x}}{|1-x|^{m+1}}
&=
\frac{1-\bar{x}}
{\left[(1-x)(1-\bar{x})\right]^{(m+1)/2}}
\nonumber\\
&=
\frac{1}{(1-x)^{(m+1)/2}}
\frac{1}{(1-\bar{x})^{(m-1)/2}}
\nonumber\\
%\label{cauchykernel1}\\
&=
\sum_{k=0}^{\infty}
\left(
\sum_{s=0}^{k}
\frac{\left(\frac{m+1}{2}\right)_{k-s}}{(k-s)!}
\frac{\left(\frac{m-1}{2}\right)_s}{s!}
\right)
x^{k-s}\bar{x}^s
\nonumber\\
&=
\sum_{k=0}^{\infty}
\frac{(m)_k}{k!}
P_k^m(x,\bar{x}).
\label{cauchykernel2}
\end{align}
Here, the Chu-Vandermonde identity \eqref{Vandermonde}, with
\[
a=\frac{m+1}{2},
\qquad
b=\frac{m-1}{2},
\]
gives $a+b=m$, which accounts for the factor $(m)_k$ in the last line of \eqref{cauchykernel2}. This leads to the definition
\begin{equation}\label{genAppelpoly}
P_k^m(x,\bar{x})
=
\sum_{s=0}^{k}
T_s^k(m)x^{k-s}\bar{x}^{s},
\end{equation}
where
\begin{equation}\label{gen Pascal}
T_s^k(m)
=
\binom{k}{s}
\frac{
\left(\frac{m+1}{2}\right)_{k-s}
\left(\frac{m-1}{2}\right)_s
}
{(m)_k}.
\end{equation}
Thus, for each fixed degree $k$, the coefficients $T_s^k(m)$ determine the expansion of $P_k^m(x,\bar{x})$ in the bivariate monomial basis
\[
x^k,x^{k-1}\bar{x},\ldots,\bar{x}^k.
\]
For brevity, we refer to the coefficients $T_s^k(m)$ as the \textit{Hypercomplex Appell Generated Numbers} (HAGN).

The preceding construction yields the following generating function.

\begin{proposition}\cite{Cacao-2017}
The hypercomplex generalized Appell polynomials satisfy
\begin{equation}\label{GenerGeomSeries:2}
g(x)
=
(1-x)^{-1}|1-x|^{1-m}
=
\sum_{k=0}^{\infty}
\frac{(m)_k}{k!}P_k^m(x,\bar{x}),
\qquad |x|<1.
\end{equation}
\end{proposition}

\begin{remark}\label{rem:analogy}
The analogy with the classical generalized geometric series
\begin{equation}\label{geom}
\frac{1}{(1-z)^m}
=
\sum_{k=0}^{\infty}
\frac{(m)_k}{k!}z^k,
\qquad z\in\mathbb C,\quad m>0,
\end{equation}
is transparent from \eqref{GenerGeomSeries:2}. The coefficients $(m)_k/k!$ are preserved, while the monomials $z^k$ are replaced by the hypercomplex Appell polynomials $P_k^m(x,\bar{x})$. In this sense, $P_k^m(x,\bar{x})$ provide the natural higher-dimensional generalization of the power functions $z^k$, extending the classical setting $\mathbb C\cong\mathbb R^2$ to $\mathbb R^{m+1}$.
\end{remark}

%%%%%%%%%%%%%%%%%%%%%%%%%%%%%%%%%%%%%%%%%%%%%%%%%%%%%%%%%%%%%%%%%%%%%%%%%%%%%%%%%%%%%%%%%%

\section{Non-symmetric generalized Pascal triangles}\label{app2}

When arranged by increasing degree $k$ for a fixed value of the parameter $m$, the HAGN coefficients in \eqref{gen Pascal} form a one-parameter family of non-symmetric Pascal triangles, whose entries are, in general, rational rather than integer numbers. Basic properties of these triangles were established by Falcão and Malonek \cite{Falcao-2012} using combinatorial methods, including recurrence relations. In this appendix, we collect some further properties of the coefficients that are relevant to the discussion in Section~\ref{sec:triang}. In particular, we establish their partition-of-unity property, examine their restrictions to the real and complex cases, and describe the integer-valued scaling that arises for odd values of the parameter.

The coefficients are explicitly given by (see \eqref{gen Pascal})
\begin{equation}\label{tks}
T_s^k(m)
=
\binom{k}{s}
\frac{
\left(\frac{m+1}{2}\right)_{k-s}
\left(\frac{m-1}{2}\right)_s
}{
(m)_k
},
\qquad
m,k=1,2,\ldots,\quad s=0,1,\ldots,k.
\end{equation}

We first consider the restriction of the hypercomplex Appell polynomials to the real axis. Setting $\underline{x}=0$ in the paravector $x$ gives
$x=\bar{x}=x_0$. Hence, from \eqref{genAppelpoly},
\begin{equation}\label{eval:1}
P_k^m(x_0)
=
x_0^k\sum_{s=0}^k T_s^k(m),
\qquad k=0,1,\ldots.
\end{equation}
Thus, the value of the sum of the coefficients for each fixed degree $k$ determines whether the restriction of $P_k^m$ to the real axis coincides with the ordinary monomial $x_0^k$. Using the Chu-Vandermonde identity \eqref{Vandermonde} with
\[
a=\frac{m+1}{2},
\qquad
b=\frac{m-1}{2},
\]
we obtain
\begin{equation}\label{eval:2}
\sum_{s=0}^k T_s^k(m)
=
\frac{1}{(m)_k}
\sum_{s=0}^k
\binom{k}{s}
\left(\frac{m+1}{2}\right)_{k-s}
\left(\frac{m-1}{2}\right)_s
=
1.
\end{equation}
Consequently, for every fixed degree $k$, the coefficients
$T_s^k(m)$ form a partition of unity. Substitution into \eqref{eval:1} then gives
\[
P_k^m(x_0)=x_0^k,
\]
independently of the value of $m\geq1$. Thus, restriction from the hypercomplex setting in $\mathbb{R}^{m+1}$ to the real axis yields precisely the ordinary monomial basis.

The complex case is recovered when $m=1$. Writing
$x=x_0+x_1e_1=z$, with $e_1=i$, we have
$\mathbb{R}^{m+1}=\mathbb{R}^2\cong\mathbb{C}$ and
$(1)_k=k!$. Moreover,
\[
T_0^k(1)=1,
\qquad
T_s^k(1)=0,\quad 1\leq s\leq k,
\]
so that
\[
P_k^1(z,\bar z)=z^k.
\]
Hence the hypercomplex Appell polynomials consistently recover the usual real and complex power functions under the corresponding lower-dimensional restrictions. This provides an additional confirmation of the analogy described in Remark~\ref{rem:analogy}.

We next consider the factor
\[
\frac{(m)_k}{k!}
\]
appearing in the generating function \eqref{GenerGeomSeries:2}. For positive integer $m$ and $k\geq0$, the Pochhammer symbol satisfies
\begin{equation}\label{Pochtrans}
\frac{(m)_k}{k!}
=
\binom{m+k-1}{k},
\quad\text{or equivalently}\quad
(m)_k
=
k!\binom{m+k-1}{k}.
\end{equation}
This identity is useful for converting the Pochhammer representation of the HAGN coefficients into binomial form.

The same quantity also has a natural combinatorial interpretation. If
$\mathcal{H}_k(\mathbb{R}^m)$ denotes the space of homogeneous polynomials of degree $k$ in $m$ variables, then
\[
\dim\mathcal{H}_k(\mathbb{R}^m)
=
\binom{m+k-1}{k}
=
\frac{(m)_k}{k!}.
\]
Thus, the coefficient multiplying $P_k^m$ in the generating function has, besides its role in the expansion of the generalized Cauchy-kernel, a direct interpretation as the dimension of the space of homogeneous polynomials of degree $k$ in $m$ variables.

Finally, we recall the integer-valued scaling of the HAGN coefficients for odd values of the parameter, as discussed in \cite{Cacao-2023}. Let
\[
m=2n+1,
\qquad n\geq1.
\]
Using \eqref{Pochtrans} to rewrite the Pochhammer symbols in \eqref{tks} in terms of binomial coefficients gives
\begin{equation}\label{tksodd}
T_s^k(2n+1)
=
\frac{
\binom{n+k-s}{n}
\binom{n+s-1}{n-1}
}{
\binom{k+2n}{2n}
}.
\end{equation}
The denominator is independent of $s$. Therefore, multiplying the entries in the $k$-th row by the common factor
$\binom{k+2n}{2n}$ produces integer-valued coefficients. We refer to these as the {\it Scaled HAGN} (SHAGN). Explicitly,
\begin{equation}\label{tksscaled}
\binom{k+2n}{2n}T_s^k(2n+1)
=
\binom{n+k-s}{n}
\binom{n+s-1}{n-1}.
\end{equation}
Thus, for odd $m$, the rational HAGN entries can be transformed into a non-symmetric Pascal triangle with integer entries.

A particularly useful form is obtained by fixing an integer $N$ and setting
\[
k=N-2n.
\]
Since $k+2n=N$, \eqref{tksscaled} becomes
\begin{equation}\label{SHAGN}
\binom{N}{2n}
T_s^{N-2n}(2n+1)
=
\binom{N-n-s}{n}
\binom{n+s-1}{n-1}.
\end{equation}
This is precisely the binomial-product representation used in
\eqref{bin}. It provides the integer-valued form of the HAGN coefficients that underlies the combinatorial arrays considered there.

%%%%%%%%%%%%%%%%%%%%%%%%%%%%%%%%%%%%%%%%%%%%%%%%%%%%%%%%%
\section*{Acknowledgments}
The first author is a member of the \textit{Gruppo Nazionale per il Calcolo Scientifico} of the Istituto Nazionale di Alta Matematica. Her research was partially supported by the  INdAM--GNCS Project, code CUP$\_$E53C25002010001.
The work of the third author at CIDMA was supported by FCT (Portuguese Foundation for Science and Technology) under Projects 
\noindent UID/04106/2025 (https://doi.org/10.54499/UID/04106/2025) and 
\noindent UID/PRR/04106/2025 \\(https://doi.org/10.54499/UID/PRR/04106/2025).

%%%%%%%%%%%%%%%%%%%%%%%%%%%%%%%%%%%

%%       \cite{<label>} ==> [1]

%%


\begin{thebibliography}{99}

\bibitem{Aceto-2026} 
Lidia Aceto, Pietro Antonio Grassi,  Helmuth Robert Malonek,
\textit{Clifford-Appell formulation of a Dirac-type Kronig-Penney model in condensed matter physics}, submitted, 2026.

\bibitem{BDS-1982} Fred Brackx, Richard Delanghe, Frank Sommen,  \textit{Clifford analysis}. Pitman, Boston-London-Melbourne, 1982.

\bibitem{Cacao-2017} 
Isabel Cação, Helmuth Robert Malonek, Maria Irene Falcão, \textit{Hypercomplex Polynomials, Vietoris’ Rational Numbers and a Related Integer Numbers Sequence.} Complex Anal. Oper. Theory 11, 1059–1076, 2017.

\bibitem{Cacao-2023} 
Isabel Cação, Helmuth Robert Malonek, Maria Irene Falcão, Graça Tomaz,   \textit{Intrinsic properties of a non-symmetric number triangle.} J. Integer Seq., 26(4): Art. 23.4.8, 12, 2023.

\bibitem{Falcao-2012} 
Maria Irene Falcão,  Helmuth Robert Malonek,  \textit{A note on a one-parameter family of non-symmetric number triangles.} Opuscula Math., 32(4): 661--673, 2012.

\bibitem{FGS-2025} 
Joaquín Figueroa,  Ivan Gonzalez,  Daniel Salinas-Arizmendi,
 \textit{A novel transfer matrix framework for multiple Dirac delta potentials}, Phys. Lett. A, 555: Paper No. 130785,  2025. 

\bibitem{Malonek-2006} 
 Helmuth Robert Malonek, Maria Irene Falcão, \textit{3D-mappings using monogenic functions}, ICNAAM-2006 Conference Proceedings, Wiley-VCH, Weinheim, 615--619, 2006.
 


\end{thebibliography}



\end{document}